\documentclass{article}
\usepackage[a4paper]{geometry}
\usepackage{authblk}
\usepackage{amsmath}
\usepackage{amssymb}
\usepackage{latexsym}
\usepackage{amsthm} 
\usepackage{array}
\usepackage{paralist}
\usepackage{algorithm, algorithmic}
\usepackage{graphicx} 
\usepackage{subfig}
\usepackage{multirow} 
\usepackage{tikz}
\usetikzlibrary{shapes,arrows,positioning,chains,decorations.pathreplacing,calc,patterns}

\newtheorem{definition}{Definition}

\newtheorem{proposition}{Proposition}

\newcommand{\set}[1]{\ensuremath{\left\{#1\right\}}} 

\newcommand{\hd}[1]{\mathsf{H}(#1)}
\newcommand{\tl}[1]{\mathsf{S}(#1)}
\newcommand{\comp}{\mathbin{;}}

\newcommand{\pa}{\mathit{PA}}
\newcommand{\seen}{\mathit{SEEN}}
\renewcommand{\mp}{\mathit{MP}}
\newcommand{\pms}{\mathit{PMS}}
\newcommand{\crs}{\mathit{CRS}}
\newcommand{\pd}{\mathit{PD}}

\begin{document}
%

\title{Path Conditions and Principal Matching:\\ A New Approach to Access Control}

\author{Jason Crampton}
\author{James Sellwood}
\affil{Information Security Group\\ Royal Holloway, University of London}


\maketitle
\begin{abstract}
Traditional authorization policies are user-centric, in the sense that authorization is defined, ultimately, in terms of user identities.
We believe that this user-centric approach is inappropriate for many applications, and that what should determine authorization is the relationships that exist between entities in the system.
While recent research has considered the possibility of specifying authorization policies based on the relationships that exist between peers in social networks, we are not aware of the application of these ideas to general computing systems.
We develop a formal access control model that makes use of ideas from relationship-based access control and a two-stage method for evaluating policies.
Our policies are defined using path conditions, which are similar to regular expressions.
We define semantics for path conditions, which we use to develop a rigorous method for evaluating policies.
We describe the algorithm required to evaluate policies and establish its complexity.
Finally, we illustrate the advantages of our model using an example and describe a preliminary implementation of our algorithm.
\end{abstract}

%
%
%

\section{Introduction}
Access control is an essential security service in any multi-user computer system.
It provides a mechanism by which different users are restricted in the actions they can perform within the system.
An access control service typically comprises a policy decision point and a policy.
An attempt by a user to interact with a system resource, usually known as an authorization request, is evaluated by the policy decision point and is only permitted if that interaction is authorized by the policy.

An access control model provides a syntax for authorization policies and a specification of the algorithm used by the policy decision point to evaluate requests.
Many access control models focus on the user and authorizing the user to perform particular actions.
As is customary in the literature, we will use the terms \emph{subjects} and \emph{objects} when referring to the parties who are to, respectively, perform and be the target of authorization (inter)actions.

Access control has been the subject of significant research and development in the last 40 years.
As our use of technology and the connectivity of our devices has increased, the need for ever more robust and scalable access control models has also grown.
New models attempt to improve on the failings of their predecessors, and often do so by redefining the policy foundations upon which authorization decisions are made.
The protection matrix model, for example, simply enumerated all authorized actions.
While this provides for precise specification of authorization policies, it does not scale well and is difficult to manage.
In order to ease this administrative burden, various improvements have been employed by modern operating systems.
The Unix operating system, for example, replaces the individual subjects with a mapping, performed at the time of request evaluation, to one of three security principals (owner, group and world)~\cite{Crampton_UnixAccessControl}.
In this way, whilst each object must still be enumerated, the enumeration of subjects is limited to just these three security principals.
With complex systems involving numerous users, this design dramatically reduces the space and administrative complexity of the underlying policy.
However, it also greatly reduces the flexibility afforded when compared with defining authorization at the user level.

Role-based access control (RBAC), which is widely used and has been the subject of extensive research in recent years, assigns a user to one or more organizational roles.
These roles are then authorized to perform certain actions on particular resources.
These roles, which are defined on a per-system basis, thus reduce the administrative burden of the protection matrix (assuming the number of roles is significantly less than the number of users), and provide a level of flexibility not available within the Unix model.
This increased flexibility also restores some of the clarity that was lost when users where abstracted behind Unix's three, very general, security principals.

A significant disadvantage with RBAC is that it takes no account of the specific relationship that might exist between a user and the resource for which access is requested.
Thus every user assigned to a doctor role can access all electronic health records if the doctor role is authorized to do so.
Clearly, it would more appropriate if the only users that are authorized to access a particular health record have a specific relationship with the subject of that record.
In short, RBAC is not as ``fine-grained'' as its supporters claim.
RBAC models that use private or parameterized roles have been introduce to tackle these kinds of problems~\cite{Ge_ParamRoles,Giuri_RT,Sandhu_RBAC}.
However, this often leads to a proliferation of roles that undermines the advantages provided by the basic RBAC model (as the number of roles tends towards the number of users).
Thus, we believe a new approach is required: an approach that combines the scalability of RBAC with the granularity of the protection matrix model and permits the specification of authorization rules on a per user-resource basis.

Recent research on access control in social networks has used the (social) relationship(s) that exist between users in such networks as the basis for specifying authorization rules~\cite{Carminati_Enforcing,ChengPS12passat,ChengPS12dbsec,Fong_ReBAC}.
The relationship information available in social networks provides additional context from which access control decisions can be derived.
We believe that relationship-based access control could be applied in many other scenarios.
In particular, the coarse-grained decision-making in RBAC can be refined using such relationship information.

In this paper, therefore, we develop a novel access control model in which policies are specified in terms of path conditions.
To a crude approximation our model takes inspiration from three sources: the overall design of the decision algorithm is similar to Unix; the path conditions are similar in spirit to some of the proposals for relationship-based access control; and the use of implementation-specific authorization principals bears some resemblance to RBAC.
We believe our path conditions provide a more rigorous foundation for access control mechanisms than existing proposals for relationship-based access control.
We also believe our use of authorization principals provides the desired scalability.


Our model introduces several novel contributions, the most significant being a generic model for access control systems using relationships that is not limited to social networks but can be used to describe access control within more traditional and more diverse environments.
Our support for logical entities, as well as the more usual users and resources, allows for a fine grained definition of authorization capable of taking into consideration relevant contextual information encoded in the relationships a request's participants have with other entities.
This is balanced with our abstraction of authorization policy to principals rather than subjects, allowing a scalable system which remains powerful and expressive.

In the next section, we describe our model for access control.
This section includes the definitions of path conditions, principal-matching rules and authorization policies, and an explanation of how requests are evaluated.
In Section~\ref{sec:Algorithm}, we consider the algorithm for matching principals in more detail, presenting a pseudo-code listing, an analysis of the algorithm's complexity and a description of a preliminary implementation in Python.
We also describe the results of some simple experiments.
We then compare our model to existing, related work and conclude the paper with a summary of our contributions and ideas for future work.
The appendix includes an extended example, fragments of which are used throughout the paper.
This extended example is used in our experiments.

\section{The Authorization Model}\label{sec:AuthZModel}


Informally our model is based on the idea of a labelled graph, in which nodes represent entities within the system and edges represent relationships between entities.
Nodes may represent concrete entities, such as users and resources, or logical entities, with which other entities are associated.
The relationships' labels are used to define path conditions which can be matched by chains of edges within the graph.
A path condition, essentially, identifies a set of authorization principals that is associated with a request.
Those principals are authorized to perform actions, thus determining whether a request is authorized or not.
Thus our model uses a two-stage decision process: we first identify the principals relevant to the request and then determine whether those principals are authorized.

As we allow entities of various types within our graph, we can make use of a variety of kinds of relationship when processing the authorization decision.
If we were to solely include users within our graph, then it could mimic a social network and would be limited to inter-personal relationships for access control policy definition.
By including group and resource entity types, we expand the possible subjects and objects, and also the possible relationships which can inform authorization decisions.
In Section~\ref{sec:AuthZModel:Ex}, we show that the RBAC model can be seen as an instance of our model.
If, however, we include additional entity types and relationships then we can make more fine-grained decisions, as illustrated by the extended example in Appendix~\ref{appendix:section:example}.

\subsection{The System Model}

Formally, we assume the existence of a set of system entities, which includes the sets of subjects and objects.
Each entity has a type and relationships may exist between certain types of entities.
Some relationships, such as \textsf{Sibling-of}, are symmetric, while others, such as \textsf{Brother-of}, are not.
A system model defines the types along with the entity relationships that are permitted.

\begin{definition}
    A \emph{system model} comprises a set of types $T$, a set of relationship labels $R$, a set of \emph{symmetric} relationship labels $S \subseteq R$ and a \emph{permissible relationship graph} $G_{\textrm{PR}} = (V_{\textrm{PR}},E_{\textrm{PR}})$, where $V_{\textrm{PR}} = T$ and $E_{\textrm{PR}} \subseteq T \times T \times R$.
\end{definition}

The example in Appendix~\ref{appendix:section:example} defines a number of types, including \textsf{Group}, \textsf{Project} and \textsf{User}, and the relationship type \textsf{Client-of}.
Part of the permissible relationship graph includes the edges $(\textsf{Group}, \textsf{Project})$ and $(\textsf{Group}, \textsf{User})$, both labelled with the \textsf{Client-of} relationship.
Figure~\ref{img:example_company_edge_constraints} (in the appendix) defines the entire permissible relationship graph.

\begin{definition}
    Given a system model $(T,R,S,G_{\textrm{PR}})$, a \emph{system instance} is defined by a \emph{system graph} $G = (V,E)$ where $V$ is the set of entities and $E \subseteq V \times V \times R$.
    We say $G$ is \emph{well-formed} if for each entity $v$ in $V$, $\tau(v) \in T$, and for every edge $(v,v',r) \in E$, $(\tau(v),\tau(v'),r) \in E_{\textrm{PR}}$.
\end{definition}

The system model constrains the `shape' of the system graph by restricting the edges that can be specified.
Note that we may have multiple edges between two entities in our system graph (because two or more relationships may exist between vertices).
Such a graph is sometimes called a \emph{multigraph}.
We will depict an edge $(v,v',s)$, when $s \in S$, without arrowheads, as can be seen in Figure~\ref{img:relationship_examples} in the case of the \textsf{Sibling-of} relationship.
(Due to the symmetry of $s$, the edge $(v,v',s)$ implies an edge $(v',v,s)$ and vice versa.)
An edge $(v,v',r)$, when $r \in R \setminus S$, is directed from $v$ to $v'$, depicted with an arrowhead pointing towards $v'$ (see Figure~\ref{img:relationship_examples:a}).
The directed edges $(v,v',r)$ and $(v',v,r)$ represent two different relationships. Of course both may belong to $E$, in which case this will be depicted with arrowheads at both ends of the link between $v$ and $v'$ (see Figure~\ref{img:relationship_examples:b}).

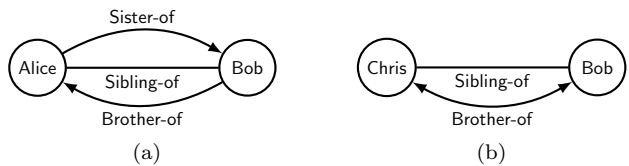
\begin{figure}[!ht]\centering
    \subfloat[]{
        \begin{tikzpicture}
            [node distance=3cm and 2cm,
            every circle node/.style={draw,minimum width=30pt},thick,
            every node/.append style={scale=0.7, transform shape}]
            \begin{scope}[>=latex] 
                \node[circle] (a) {\textsf{Alice}};
                \node[circle,right=of a] (b) {\textsf{Bob}};
                \draw[thick] (a) to node[auto,swap] {\textsf{Sibling-of}} (b);
                \draw[<-,thick] (a) to[color=black,bend right] node[auto,swap] {\textsf{Brother-of}} (b);
                \draw[->,thick] (a) to[color=black,bend left] node[auto] {\textsf{Sister-of}} (b);
            \end{scope}
        \end{tikzpicture}
        \label{img:relationship_examples:a}
    }\qquad
     \subfloat[]{
        \begin{tikzpicture}
            [node distance=3cm and 2cm,
            every circle node/.style={draw,minimum width=30pt},thick,
            every node/.append style={scale=0.7, transform shape}]
            \begin{scope}[>=latex] 
                \node[circle] (a) {\textsf{Chris}};
                \node[circle,right=of a] (b) {\textsf{Bob}};
                \draw[thick] (a) to node[auto,swap] {\textsf{Sibling-of}} (b);
                \draw[<->,thick] (a) to[color=black,bend right] node[auto,swap] {\textsf{Brother-of}} (b);
            \end{scope}
        \end{tikzpicture}
        \label{img:relationship_examples:b}
    }
    \caption{Illustrating different edges in the system graph}\label{img:relationship_examples}
\end{figure}

Figure~\ref{img:example_company_system_graph} (in the appendix) depicts a system graph containing a substantial number of nodes of different types and the relationships that exist between those nodes.
Figure~\ref{img:system-graph-fragment} shows a simple example of a system graph for illustrative purposes, based on the one in the appendix.
Users (such as $U_1$) are associated with projects ($P_1$); documents ($D_1$ and $D_2$) are grouped together in folders ($F_1$ and $F_2$) and allocated to one or more projects, either as part of a group or as a single resource.
Relationships include \textsf{Participant-of}, \textsf{Supervises}, \textsf{Resource-for}, and \textsf{Member-of}.

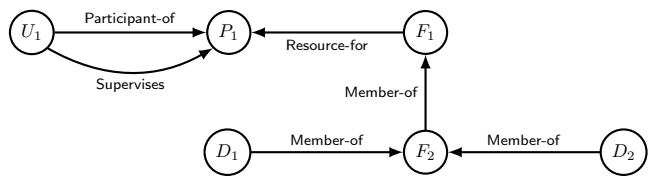
\begin{figure}[h]\centering
  \begin{tikzpicture}
      [node distance=1cm and 2cm,
      every circle node/.style={draw,minimum width=20pt},thick,
      every node/.append style={scale=0.7, transform shape}]
      \begin{scope}[>=latex] 
	  \node[circle] (u1) {$U_1$};
	  \node[circle,right=of u1] (p1) {$P_1$};
	  \node[circle,right=of p1] (f1) {$F_1$};
	  \node[circle,below=of f1] (f2) {$F_2$};
	  \node[circle,left=of f2] (d1) {$D_1$};
	  \node[circle,right=of f2] (d2) {$D_2$};
	  \draw[thick,->] (u1) to node[auto] {\textsf{\footnotesize Participant-of}} (p1);
	  \draw[thick,->] (u1.south east) to [bend right] node[auto,swap] {\textsf{\footnotesize Supervises}} (p1.south west);
	  \draw[thick,->] (f1) to node[auto] {\textsf{\footnotesize Resource-for}} (p1);
	  \draw[thick,->] (f2) to node[auto] {\textsf{\footnotesize Member-of}} (f1);
	  \draw[thick,->] (d1) to node[auto] {\textsf{\footnotesize Member-of}} (f2);
	  \draw[thick,->] (d2) to node[auto,swap] {\textsf{\footnotesize Member-of}} (f2);
      \end{scope}
  \end{tikzpicture}
\caption{A fragment of a system graph}\label{img:system-graph-fragment}
\end{figure}


\subsection{Path Conditions}\label{sec:AuthZModel:PathConds}

We use path conditions to match requests to principals (described in Section~\ref{sec:AuthZModel:ReqEval}).
In this section, we define the syntax and semantics of path conditions, and establish some basic properties of path conditions, thereby allowing us to restrict our attention to simple path conditions.

\begin{definition}\label{def:path-condition}
   Given a set of relationships $R$, we define a \emph{path condition} recursively:
    \begin{itemize}
	\item $\diamond$ is a path condition;
        \item $r$ is a path condition, for all $r \in R$;
        \item if $\pi$ and $\pi'$ are path conditions, then $\pi \comp \pi'$, $\pi^+$ and $\overline{\pi}$ are path conditions.
    \end{itemize}
   A path condition of the form $r$ or $\overline{r}$, where $r \in R$, is said to be an \emph{edge condition}.
\end{definition}

Informally, $\pi \comp \pi'$ represents the concatenation of two path conditions; $\pi^+$ represents one or more occurrences, in sequence, of $\pi$; and $\overline{\pi}$ represents $\pi$ reversed.
We define $\diamond$ for completeness.
We note that individual edge conditions could be encoded using attribute-based access control (ABAC) but it is hard to see how ABAC could be easily employed to encode longer chains of relationships.

\begin{definition}\label{def:path-condition-semantics}
 Given a system graph $G = (V,E)$ and $u,v \in V$, we write $G,u,v \models \pi$ to denote that  $G$, $u$ and $v$ \emph{satisfy path condition} $\pi$.
 Formally, for all $G,u,v,\pi,\pi'$:
    \begin{itemize}
        \item $G,u,v \models \diamond$ iff $v = u$;
        \item $G,u,v \models r$ iff $(u,v,r) \in E$;
        \item $G,u,v \models \pi \comp \pi'$ iff there exists $w \in V$ such that $G,u,w \models \pi$ and $G,w,v \models \pi'$;
        \item $G,u,v \models \pi^+$ iff $G,u,v \models \pi$ or $G,u,v \models \pi \comp \pi^+$;
        \item $G,u,v \models \overline{\pi}$ iff $G,v,u \models \pi$.
    \end{itemize}
\end{definition}

Note that an edge condition is satisfied by nodes that are adjacent in the system graph.
We use $\diamond$ to identify an empty path condition, which is of particular use in our path-matching algorithm in Section~\ref{sec:Algorithm:Description}.

In the context of the graph in Figure~\ref{img:system-graph-fragment}, for example, we have $G,U_1,F_1 \models \textsf{Participant-of} \comp \overline{\textsf{Resource-for}}$ since $G,U_1,P_1 \models \textsf{Participant-of}$ and $G,F_1,P_1 \models \textsf{Resource-for}$.

\begin{definition}\label{def:path-condition-equivalence}
 Path conditions $\pi$ and $\pi'$ are said to be \emph{equivalent}, denoted $\pi \equiv \pi'$, if, for all system graphs $G = (V,E)$ and all $u,v \in V$ we have
  \[
   G,u,v \models \pi \quad\text{if and only if} \quad G,u,v \models \pi'.
  \]
\end{definition}

\begin{proposition}\label{pro:simple-equivalences}
For all path conditions $\pi_1$ and $\pi_2$:
   \begin{inparaenum}[\em (i)]
    \item $\pi_1 \equiv \pi_1 \comp \diamond \equiv \diamond \comp \pi_1$
    \item $\overline{\diamond} \equiv \diamond$
    \item $\overline{\pi_1 \comp \pi_2} \equiv \overline{\pi_2} \comp \overline{\pi_1}$
    \item $\overline{\pi_1^+} \equiv \overline{\pi_1}^+$.
   \end{inparaenum}
\end{proposition}


\begin{proof}
 All results follow immediately from Definitions~\ref{def:path-condition-semantics} and~\ref{def:path-condition-equivalence}.
 Consider (iii), for example.
 By definition, $G,u,v \models \overline{\pi_1 \comp \pi_2}$ if and only if $G,v,u \models \pi_1 \comp \pi_2$.
 And $G,v,u \models \pi_1 \comp \pi_2$ if and only there exists $w$ such that $G,v,w \models \pi_1$ and $G,w,u \models \pi_2$.
 Thus we have $G,u,v \models \overline{\pi_1 \comp \pi_2}$ if and only if there exists $w$ such that $G,w,v \models \overline{\pi_1}$ and $G,u,w \models \overline{\pi_2}$.
 That is $G,u,v \models \overline{\pi_2} \comp \overline{\pi_1}$.
\end{proof}

\begin{definition}
   Given a set of relationships $R$, we define a \emph{simple path condition} recursively:
    \begin{itemize}
	\item $\diamond$, $r$ and $\overline{r}$, where $r \in R$, are simple path conditions;
        \item if $\pi \ne \diamond$ and $\pi' \ne \diamond$ are simple path conditions, then $\pi \comp \pi'$ and $\pi^+$ are simple path conditions.
    \end{itemize}
\end{definition}

In other words, $\overline{\star}$ occurs in a simple path condition if and only if $\star$ is an element of $R$.
It follows from Proposition~\ref{pro:simple-equivalences} that every path condition may be reduced to a simple path condition.
The path condition $\overline{\overline{r_1 \comp r_2} \comp (r_1 \comp r_3)^+}$, for example, can be transformed into the equivalent path condition $(\overline{r_3} \comp \overline{r_1})^+ \comp r_1 \comp r_2$ using the equivalences in Proposition~\ref{pro:simple-equivalences}.
Henceforth, we assume all path conditions are simple.

\subsection{Policy Specification}\label{sec:AuthZModel:PolicySpec}

Subjects within a system request authorization to perform actions on objects. The policies of a system define the authorized and unauthorized actions and the rules for determining the principals to which these actions are assigned. Principals are mapped to paths within the system graph, where these paths exist between the subject and object of an authorization request. The potential paths are described by path conditions, which are defined using relationships.

\begin{definition}
     Let $P$ be a set of \emph{authorization principals} and let $R$ be a set of relationship labels.
     A \emph{principal-matching rule} has the form $(\pi, p)$, where $p$ is an authorization principal and $\pi$ is either a path condition defined on $R$ or the special symbol $\top$.
     The path condition $\pi$ is called a \emph{principal-matching condition}. A \emph{principal-matching policy} $\rho$ is a list of principal-matching rules.
\end{definition}

Informally, a principal-matching rule $(\pi,p)$ is applicable to a request $(s,o,a)$ if there is a path from $s$ to $o$ in the system graph that satisfies $\pi$.

In order to support scenarios where a default principal should apply, much like the concept of `world' in the Unix access control system, we allow the definition of a special principal-matching rule with the principal-matching condition set to $\top$.
This default principal-matching rule is, if present, always the last rule in the principal-matching policy and (whenever it is evaluated) is applicable to every request.
This rule's associated principal, therefore, matches whenever the rule is evaluated.

\begin{definition}
    An \emph{authorization rule} has the form $(p,\star,a,b)$ or $(p,o,a,b)$, where $a$ is an action, $p$ is a principal, $o$ is an object and $b \in \set{0,1}$.
    An \emph{authorization policy} is a list of authorization rules.
\end{definition}

A rule of the form $(p,o,a,0)$ asserts that $p$ is explicitly \emph{unauthorized} (or \emph{prohibited}) to perform action $a$ on object $o$, while the rule $(p,o,a,1)$ explicitly authorizes $p$.
Rules of this form allow us to specify on a per-object basis the actions for which a principal $p$ is (un)authorized.
A rule of the form $(p,\star,a,0)$ asserts that the principal is unauthorized for all objects, while $(p,\star,a,1)$ asserts that the principal is authorized for all objects.
Rules of this form allow us to specify the actions for which a principal is (un)authorized, irrespective of the object to which access is requested.
In this case, the authorization policy is concentrated in the principal-matching rule.
Note also that we can combine rules $(p,\star,a,0)$ and $(p,o,a,1)$, for example, to specify that action $a$ is generally unauthorized for principal $p$, but is, as an exception, authorized for object $o$.

Table~\ref{tbl:example_company_principal_matching_policy} (in the appendix) lists the principal-matching rules for our example whilst Table~\ref{tbl:example_company_authorization_policy} lists the authorization rules. A combination of authorization rules has been used in Table~\ref{tbl:example_company_authorization_policy} to ensure that the \textsf{Project Resource User} is specifically unable to write to \textsf{Func.Spec.\#1} whilst other objects are writable by that principal.

A principal may be explicitly authorized or unauthorized for particular actions.
The absence of any explicit authorization rules may itself be considered an implicit authorization depending on the default behaviour of the system.
A default access control decision (allow or deny) needs to be specified in the event that no authorization rules apply to a request.
Systems may need to support \textsf{allow-by-default} when the system enters an emergency state, such as the opening of fire exit doors when there is a fire.
Other circumstances will commonly require fail-safe handling, where a \textsf{deny-by-default} strategy is implemented in order to ensure no unauthorised access is allowed. Some systems may be deemed so sensitive that there may be no conditions under which \textsf{allow-by-default} would be enabled.
In Section~\ref{sec:AuthZModel:ReqEval:Defaults} we discuss the specification of default strategies in our model.

\subsection{Request Evaluation}\label{sec:AuthZModel:ReqEval}

Our model for request evaluation is inspired by the Unix access control model and relationship-based access control models and is summarized in Figure~\ref{img:overview}.
From the Unix model, we take the idea of binding a request to a principal before computing an access control decision, which we combine with the idea of specifying authorization policies in terms of relationships.
Firstly, we use the subject and object specified in the request to compute a set of applicable principals.
Then we compute the actions for which those principals are authorized.
Finally a decision is made to allow or deny the request based on those authorizations.


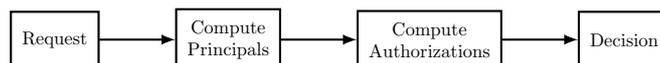
\begin{figure}[!ht]\centering
    \begin{tikzpicture}
            [node distance=3cm and 1cm,
            every rectangle node/.style={draw,minimum width=30pt, minimum height=30pt, align=center,inner sep=6pt},thick,
            every node/.append style={scale=0.7, transform shape}]
            \begin{scope}[>=latex] 
                \node[rectangle] (a) {Request};
                \node[rectangle,right=of a] (b) {Compute\\Principals};
                \node[rectangle,right=of b] (c) {Compute\\Authorizations};
                \node[rectangle,right=of c] (d) {Decision};
                \draw[->,thick] (a) to (b);
                \draw[->,thick] (b) to (c);
                \draw[->,thick] (c) to (d);
            \end{scope}
    \end{tikzpicture}
    \caption{Processing overview}\label{img:overview}
\end{figure}

We now describe request evaluation, which has two main stages and is depicted schematically in Figure~\ref{img:architecture}, in more detail.
The first stage determines a list of matched principals for the request: in Figure~\ref{img:architecture} this stage is represented by the horizontal row of steps from `START'.
%
The second stage determines the authorizations explicitly defined for those matched principals identified in the first stage.
This second stage is represented in Figure~\ref{img:architecture} by the vertical column of steps beginning at the `MP list empty?' decision point.
A conflict resolution process is employed to resolve any conflicting authorization rules and from this a decision is made.

\begin{figure*}[!th]\centering
  \begin{tikzpicture}
        [auto,
        input/.style={rectangle, text width=5em, align=center, minimum height=4em}, 
        decision/.style={diamond, draw=blue, thick, fill=blue!20, text width=4.5em, align=flush center, inner sep=1pt}, 
        block/.style ={rectangle, draw=blue, thick, fill=blue!20, text width=5em, align=center, rounded corners, minimum height=4em}, 
        start/.style ={ellipse, draw=teal, thick, fill=teal!20, text width=5.5em, align=center, minimum height=3em}, 
        end/.style ={ellipse, draw=red, thick, fill=red!20, text width=5.5em, align=center, minimum height=3em}, 
        list/.style ={rectangle, draw, thick, text width=5em, align=center, rounded corners, minimum height=4em}, 
        strategy/.style ={rectangle, draw, thick, text width=5em, align=center, rounded corners, minimum height=4em, dashed}, 
        line/.style ={draw, thick, -latex}, 
        option/.style ={dashed}, 
        every node/.append style={scale=0.65, transform shape}] 
        \matrix [column sep=5mm,row sep=5mm]
        {
            \node [input] (a1) {request $q = (s,o,a)$}; &
            \node [input] (a2) {system graph $G = (V,E)$}; &
            \node [list] (a3) {principal-matching policy $\rho$}; &
            \node [input] (a4) {default-per-subject}; &
            \node [input] (a5) {default-per-object}; &
            \node [input] (a6) {system-wide default};
            \\
            \node [start] (b1) {START}; &
            \node [block] (b2) {compute principals}; &
            \node [list] (b3) {list of matched principals $MP$}; &
            \node [decision] (b4) {$MP$ list empty?}; &
            \node [block] (b5) {process rules for no matching principal}; &
            \node [end] (b6) {authorization decision};
            \\
            &
            \node [strategy] (c2) {principal-matching strategy $PMS$}; &
            \node [list] (c3) {author\-ization policy $PA$}; &
            \node [block] (c4) {compute authorizations}; &
            \node [input] (c5) {request $q = (s,o,a)$}; &
            \\
            \node [input] (d1) {\textsf{FirstMatch}}; &
            \node [input] (d2) {\textsf{AllMatch}}; &
            &
            \node [input] (d4) {set of possible decisions $PD$}; &
            \node [input] (d5) {default per object}; &
            \node [input] (d6) {system-wide default};
            \\
            &
            \node [input] (e2) {\textsf{DenyOverride}}; &
                        \node [input] (g2) {\textsf{AllowOverride}}; &

            \node [decision] (e4) {$PD = \emptyset$?}; &
            \node [block] (e5) {process rules for no explicit permissions}; &
            \node [end] (e6) {authorization decision};
            \\
             &
            \node [input] (f2) {\textsf{FirstMatch}}; &
            \node [strategy] (f3) {conflict resolution strategy $CRS$}; &
            \node [block] (f4) {compute decision}; &
            &
            \node [end] (f6) {authorization decision};
            \\
        };
        \begin{scope}[every path/.style=line]
            \path (a1) -- (b2);
            \path (a2) -- (b2);
            \path (a3) -- (b2);
            \path (a4) -- (b5);
            \path (a5) -- (b5);
            \path (a6) -- (b5);
            \path (b1) -- (b2);
            \path (b2) -- (b3);
            \path (b3) -- (b4);
            \path (b4) -- node {Y} (b5);
            \path (b4) -- node {N} (c4);
            \path (b5) -- (b6);
            \path (c2) -- (b2);
            \path (c3) -- (c4);
            \path (c4) -- (d4);
            \path (c5) -- (b5);
            \path (c5) -- (c4);
            \path [option] (d1) -- (c2);
            \path [option] (d2) -- (c2);
            \path (d4) -- (e4);
            \path (d5) -- (e5);
            \path (d6) -- (e5);
            \path [option] (e2) -- (f3);
            \path (e4) -- node {Y} (e5);
            \path (e4) -- node {N} (f4);
            \path (e5) -- (e6);
            \path [option] (f2) -- (f3);
            \path (f3) -- (f4);
            \path (f4) -- (f6);
            \path [option] (g2) -- (f3);
        \end{scope}
    \end{tikzpicture}
    \caption{Detailed architecture}\label{img:architecture}
\end{figure*}
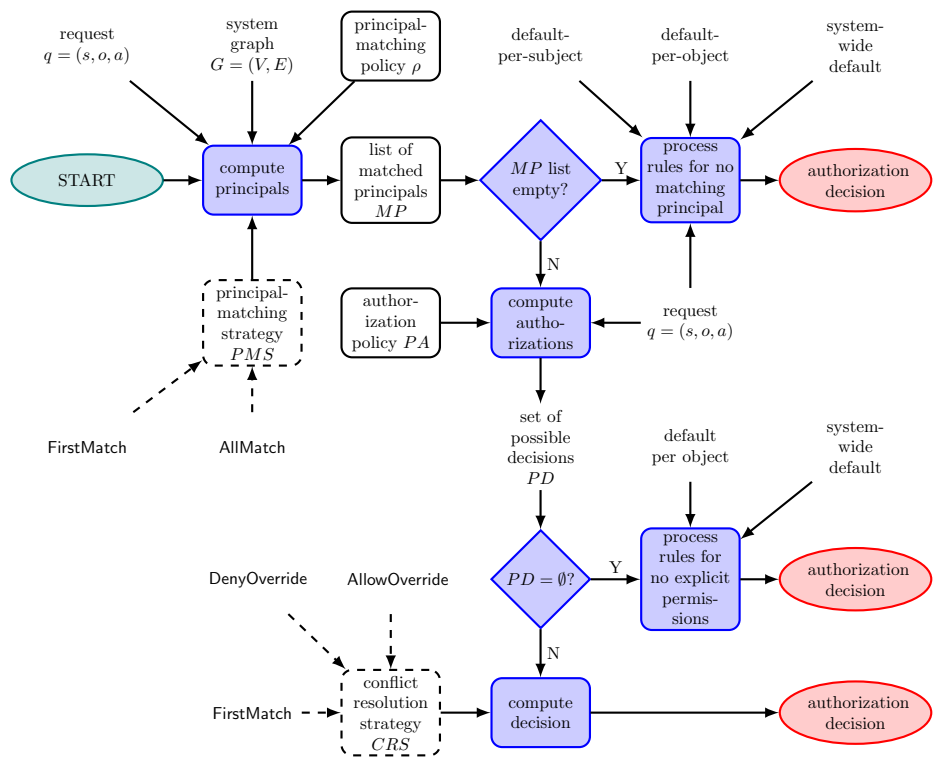

\subsubsection{Principal Matching}\label{sec:AuthZModel:ReqEval:S1}

The list of matched principals is determined by the evaluation of principal-matching rules within the principal-matching policy.
Thus, we first specify what it means for a principal-matching rule to be matched.

\begin{definition}
    Let $q = (s,o,a)$ be a request and $G = (V,E)$ be a system graph.
    Then request $q$ \emph{matches} principal-matching rule $(\pi,p)$ if $G,s,o \models \pi$.
    Given a principal-matching policy and a system graph, we write $G,q \xrightarrow{\pi} p$ if there exists a principal-matching rule $(\pi, p)$ and request $q$ matches $(\pi, p)$.
\end{definition}

Informally, a principal-matching rule maps a (complex) relationship between entities in a graph to a principal;
in other words, a principal-matching rule enables us, conceptually, to replace a path between two entities with a single edge labelled by a principal.
Figure~\ref{img:principal_matching_rule} illustrates such a matching, where request $q = (s,o,a)$ matches a principal-matching rule $(r_1 \comp \overline{r_2} \comp r_3 \comp r_4, p)$. It is worth noting that, based on the relationships shown in Figure~\ref{img:principal_matching_rule}, where $r_4$ is a symmetric label (identified by the lack of arrows on the edge), the principal-matching rule would also have been matched if the path condition had been $r_1 \comp \overline{r_2} \comp r_3 \comp \overline{r_4}$.

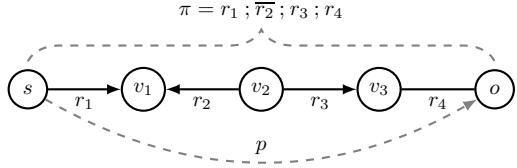
\begin{figure}[!ht]\centering
    \begin{tikzpicture}[node distance=2.5cm and 1.2cm,every circle node/.style={draw,minimum width=18pt},thick,
        every node/.append style={scale=1, transform shape},
        scale=0.8]
        \begin{scope}[>=latex] 
            \node[circle] (s) {$s$};
            \node[circle,right=of s] (v1) {$v_1$};
            \node[circle,right=of v1] (v2) {$v_2$};
            \node[circle,right=of v2] (v3) {$v_3$};
            \node[circle,right=of v3] (o) {$o$};
            \draw[->,thick] (s) to node[auto,swap] {$r_1$} (v1);
            \draw[<-,thick] (v1) to node[auto,swap] {$r_2$} (v2);
            \draw[->,thick] (v2) to node[auto,swap] {$r_3$} (v3);
            \draw[thick] (v3) to node[auto,swap] {$r_4$} (o);
            \draw[->,color=gray,dashed,thick] (s) to[color=black,bend right] node[auto] {$p$} (o);
            \draw[color=gray,decorate, decoration={brace,amplitude=12pt,raise=3pt},dashed,thick] (s.north) to node[color=black,yshift=20pt,auto] {$\pi = r_1 \comp \overline{r_2} \comp r_3 \comp r_4$} (o.north);
        \end{scope}
    \end{tikzpicture}
    \caption{Principal-matching rule}\label{img:principal_matching_rule}
\end{figure}

A request may match more than one rule in the principal-matching policy.
A \emph{principal-matching strategy} (PMS) defines how the principals in matched rules should be combined (if at all).
We consider two very natural PMSs: \textsf{FirstMatch} and \textsf{AllMatch}, but other options may be appropriate in some circumstances.
The former evaluates the list of principal-matching rules in order and terminates when a path condition is matched, returning the corresponding principal.
The latter evaluates the entire list of rules in the policy and returns a list of the principals in rules for which the request matches the path condition.


If used in conjunction with the \textsf{FirstMatch} PMS, the default principal rule $(\top,p)$, when present, would only be triggered, and so only apply, if no other rule matches.
When used with \textsf{AllMatch} this rule would always apply, resulting in the default principal always being added to the list of matched principals.

An \emph{authorization system} comprises a principal-matching policy $\rho$, a principal-matching strategy $\pms$, an authorization policy $\pa$, and a conflict resolution strategy $\crs$ (described in the next section).
%
Given an authorization system, a system graph $G$ and a request $q$, the list of \emph{matched principals} $\mp$ includes those principals resulting from successful matches made in accordance with the specified principal-matching strategy. We write $G,q \xrightarrow{\rho} \mp$ to indicate that the list of matched principals for $q$ (with respect to policy $\rho$ and system graph $G$) is $\mp$.
If $\mp$ is empty then an authorization decision must be made based on pre-defined defaults. This process is described in Section~\ref{sec:AuthZModel:ReqEval:Defaults}.

\subsubsection{Computing Authorizations and Decisions}\label{sec:AuthZModel:ReqEval:S2}
The second stage of request authorization identifies whether the requested action (on the object) is explicitly authorized or unauthorized for one or more of the matched principals. Subsequently any conflicting assignments are resolved and we determine whether the requested action should, therefore, be allowed or denied.

\begin{definition}
    Given a policy $\rho$, a request $q = (s,o,a)$ and a system graph $G$ such that $G,q \xrightarrow{\rho} \mp$, we define the set of \emph{possible decisions}, denoted $\pd$, to be $\set{b \in \set{0,1} : (p,o,a,b) \in \pa,  p \in \mp}$.
\end{definition}

$\pd$ can take one of four values: $\set{0}$, $\set{1}$, $\set{0,1}$ and $\emptyset$.
\begin{itemize}
  \item If $\pd = \set{b}$, $b \in \set{0,1}$, a decision can unambiguously be made (as a deny or allow, respectively).
  \item If $\pd = \set{0,1}$, we must employ a \emph{conflict resolution strategy} (CRS) to determine the decision.
	We define three conflict resolution strategies: \textsf{FirstMatch}, \textsf{DenyOverride} and \textsf{AllowOverride}.
    The use of one of these strategies allows a single decision to be made from the conflicting assignments.
	In order to support the first of these, we require that the set of possible decisions be determined by considering each authorization, from the list $\pa$, in turn.
	
	The \textsf{FirstMatch} CRS takes the first element to be added to $\pd$ as the decision.
	In this way, if a positive authorization is identified first, then the request is allowed.
	If a negative authorization is identified first, however, then the request is denied.
	
	The \textsf{DenyOverride} and \textsf{AllowOverride} CRSs allow their respective elements, $0$ and $1$, to take precedence over the alternative, no matter which is identified first.
  \item The final case is $\pd = \emptyset$. 
	In this case, the authorization decision must once again be made using pre-defined defaults, as explained in Section~\ref{sec:AuthZModel:ReqEval:Defaults} below.
\end{itemize}


\subsubsection{Defaults}\label{sec:AuthZModel:ReqEval:Defaults}
There are two circumstances when default decision making applies. The first is when no matched principals are identified, whilst the second is, as just described, when the set of possible decisions is empty.

To accommodate varying needs in these circumstances, we allow for default allow or deny of a request to be determined at one of the following levels: default-per-subject, default-per-object or system-wide default. We only support the default-per-subject when there are no matched principals, and not later, when there are no explicit authorizations. At the time when the set of possible decisions is determined, the subject is no longer directly relevant, having already been used to identify the appropriate matched principals. It is therefore unnecessary to reconsider the subject in order to evaluate the authorization decision.

The three defaults are evaluated in order, where specified, with the first applicable default determining the authorization decision. In this way, if there is a default specified for the subject $s$ of the request $q = (s,o,a)$, the subject's default (allow or deny) applies. If no subject default is defined for $s$, then the default for the object $o$ of the request shall apply, if specified. If there is no subject default for $s$ and no object default for $o$, then the system-wide default shall apply. Whilst defaults for the subject and object are optional and may not be specified for the entities involved in the request, a system-wide default must be specified so as to ensure authorization decisions can be made in all circumstances.

\subsection{Special Cases}\label{sec:AuthZModel:Ex}

The Unix access control mechanism employs a similar, albeit far simpler, mapping technique as that used above to identify principals from path conditions~\cite{Crampton_UnixAccessControl}. It can, therefore, be trivially represented using our model.
In particular, the system model contains a set of three types: users, groups and objects, and a set of three relationships (none of which are symmetric): \textsf{User-object}, \textsf{User-group} and \textsf{Group-object}, which we will label \textsf{uo}, \textsf{ug} and \textsf{go}, respectively.
The permissible relationship graph links the users to the objects and to the groups, as well as linking the groups to the objects (this is as our relationship naming suggests).
There are three principal-matching rules: $(\textsf{uo}, \textsf{owner})$, $(\textsf{ug} \comp \textsf{go},\textsf{group})$ and, the default, $(\top,\textsf{world})$.
Finally, we use the \textsf{FirstMatch} PMS and evaluate the rules in the above order.

Note also that we can configure our model to implement core RBAC~\cite{ANSI_RBAC}.
We assume the set of entities is the disjoint union of users, roles, permissions and objects.
Then there are two types of relationship, the \textsf{User-role} relationship, referred to as user assignment and abbreviated \textsf{ua}, along with the \textsf{Role-permission} relationship, referred to as permission assignment and abbreviated \textsf{pa}.
At its simplest we then define a principal-matching policy where each rule has the form $(\textsf{ua} \comp \textsf{pa},p)$ where the principal $p$ has the same name as the permission identified by the \textsf{pa} edge.
The authorization policy contains elements $(p,ob,op,1)$ which map the principals to objects, allowing them operations (as per the permission binary relation in RBAC).

Additionally, we can introduce the \textsf{Role-role} relationship (abbreviated \textsf{rr}) in order to extend this configuration to implement a role hierarchy.
Finally, we could also introduce the \textsf{User-permission} relationship (abbreviated \textsf{up}), in order to articulate exceptions to the basic RBAC model by directly associating permissions with users.

Our model does not directly support the concept of sessions.
However, if we were to introduce support for changing the system graph, we could employ a \textsf{User-session-role} relationship.
The \textsf{User-session-role} relationship may only connect users and roles who are already joined by a \textsf{User-role} relationship.
We then modify the original principal-matching rules to have the form $(\textsf{usr} \comp \textsf{pa},p)$.
Supporting (constrained) updates to the system graph in real time will be an important aspect of our future work.

\section{Path Matching}\label{sec:Algorithm}
Principal matching, the first stage of request evaluation, described in Section~\ref{sec:AuthZModel:ReqEval:S1}, is the most complex part of request evaluation.
(The second stage amounts to a sequence of simple lookups and comparisons.)
Principal matching requires us to determine whether there exists a path in the graph from subject to object that matches a path condition.
In this section, we describe the \textsf{MatchPrincipal} algorithm, which takes a path condition, two nodes (the subject and object of a request), the set of symmetric relationship labels and a system graph as inputs and returns a Boolean value indicating whether there exists a matching path in the graph.

The algorithm uses a (modified) breadth-first search to determine whether there exists a path in the system graph that begins at the subject and ends at the object such that concatenation of the relationship labels is equal to the path condition.
It is employed iteratively to as many rules in the principal-matching policy as required, given the PMS in use: if \textsf{FirstMatch} is used then the algorithm is run on each principal-matching rule in turn, until a match is found; if the \textsf{AllMatch} PMS is used, the algorithm is run for every rule in the policy.
In order to determine satisfaction of a simple path condition, we attempt to satisfy its component edge conditions one at a time.
It is helpful to define the head and suffix of a path condition: the head is used to match edge labels in the graph, while the suffix determines the residual path condition.

\begin{definition}\label{def:HeadAndSuffix}
Let $\pi \ne \diamond$ be a simple path condition.
Then we define the \emph{head} and \emph{suffix} of $\pi$, denoted $\hd{\pi}$ and $\tl{\pi}$, respectively, as follows:
\begin{itemize}
 \item $\hd{r} = r$ and $\tl{r} = \diamond$;
 \item $\hd{\overline{r}} = \overline{r}$ and $\tl{\overline{r}} = \diamond$;
 \item $\hd{\pi_1 \comp \pi_2} = \hd{\pi_1}$ and $\tl{\pi_1 \comp \pi_2} = \tl{\pi_1} \comp \pi_2$;
 \item $\hd{\pi^+} = \hd{\pi}$ and $\tl{\pi^+} = \tl{\pi} \comp \pi^*$, where $\pi^*$ denotes $0$ or more occurrences of $\pi$.
\end{itemize}
\end{definition}

\begin{proposition}\label{def:head-path-condition-equals-edge}
 Let $\pi$ be a simple path condition.
 Then $\hd{\pi}$ is equal to $r$ or $\overline{r}$ for some $r \in R$.
 Moreover, $\tl{\pi}$ is a simple path condition.
\end{proposition}

\begin{proof}
 The results follow immediately by a simple induction on the structure of simple path conditions.
\end{proof}

\begin{proposition}\label{pro:path-condition-equiv-head-comp-suffix}
 Let $\pi$ be a simple path condition.
 Then $\pi \equiv \hd{\pi} \comp \tl{\pi}$.
\end{proposition}

\begin{proof}
 The proof proceeds by induction on the structure of $\pi$.
 Consider the (base) case $\pi = r$. 
 Then
  \begin{align*}
   G,u,v \models \hd{r} \comp \tl{r} &\Leftrightarrow G,u,v \models r \comp \diamond \\
				     &\Leftrightarrow G,u,v \models r
  \end{align*}
 Thus $\hd{r} \comp \tl{r} \equiv r$, as required.
 We prove the case $\pi = \overline{r}$ in a similar fashion.
 Now consider $\pi = \pi_1 \comp \pi_2$ and assume the result holds for $\pi_1$ and $\pi_2$.
 Then
  \begin{align*}
     G,u,v \models \hd{\pi_1 \comp \pi_2} \comp \tl{\pi_1 \comp \pi_2} 
								       &\Leftrightarrow G,u,v \models \hd{\pi_1} \comp \tl{\pi_1} \comp \pi_2  \\ 
								       &\Leftrightarrow G,u,v \models \pi_1 \comp \pi_2 
  \end{align*}
 Finally, consider $\pi^+$ and assume the result holds for $\pi$.
 Then
  \begin{align*}
   G,u,v \models \hd{\pi^+} \comp \tl{\pi^+} 
				             &\Leftrightarrow G,u,v \models \hd{\pi} \comp \tl{\pi} \comp \pi^* \\ 
				             &\Leftrightarrow G,u,v \models \pi \comp \pi^* \\ 
				             &\Leftrightarrow G,u,v \models \pi^+
  \end{align*}
 concluding the proof.
\end{proof}

We now develop the path-matching algorithm in more detail.

\subsection{The Path-Matching Algorithm}\label{sec:Algorithm:Description}

%
The algorithm takes a start node (the subject), a target node (the object) and a path condition as part of its input.
The current node is initialized to be the start node.
The path-matching algorithm traverses the provided system graph `consuming' the head of the path condition as it matches it against (one or more of) the relationship labels associated with incident edges of the current node.
It then considers each of the adjacent edges in turn replacing the path condition with the relevant suffix.
The algorithm terminates if it `consumes' the entire path condition with the adjacent node equal to the target node or if no further matches can be made.

If we consider, for example, the graph in Figure~\ref{img:principal_matching_rule}, the request $(s,o,a)$ and path condition $r_1 \comp r_2$, then $\hd{r_1 \comp r_2} = r_1$, which is the label on edge $(s,v_1,r_1)$.
Hence, the edge is traversed and we next consider the node $v_1$ with path condition $\tl{r_1 \comp r_2} = r_2$.
The algorithm terminates at this point (returning false) because there is no outgoing edge from $v_1$ labelled $r_2$.

The \textsf{MatchPrincipal} algorithm (listed in Algorithm~\ref{alg:MatchPrincipal}) is, essentially, a modified breadth-first search algorithm.
However, there are some awkward aspects to the design of the algorithm.
First, we have to allow for nodes to be revisited.
Second, we have to allow matching of edge conditions of the form $r$ and $\overline{r}$.
Finally, our algorithm has to be able to handle path conditions of the form $\pi^+$ without entering an endless loop, in order for the algorithm to terminate.

\begin{algorithm}[!h]\footnotesize
    \caption{\textsf{MatchPrincipal}}
    \label{alg:MatchPrincipal}
    \algsetup{indent=2em}
    \begin{algorithmic}[1]
        \REQUIRE{Graph $G = (V,E)$, set of symmetric relationship labels $S$, nodes $u$ and $v$, and path condition $\pi$}
        \ENSURE{Returns true if $G,u,v \models \pi$ and false if it does not}
            \STATE{Initialize empty queue $Q$}
            \STATE{Initialize empty set of visited nodes $\seen$}
            \STATE{\textbf{add} $(u, \pi)$ \textbf{to} $Q$}
          	\STATE{$\seen = \seen \cup \set{(u, \pi)}$}
            \WHILE{$Q$ is not empty}
                \STATE{\textbf{dequeue} next entry $(h, \phi)$ from $Q$}
                \STATE{Initialize empty list of (node, suffix) tuples $\Theta$}
                \STATE{\COMMENT{consider edges directed away from $h$}}
                \FOR{each edge $(h,w,r) \in E$}
                    \IF{$\phi = \pi_1^* \comp \pi_2$}
                        \STATE{$\Theta = \Theta \sqcup [(h, \pi_2)]$}
                        \STATE{$\phi = \pi_1^+ \comp \pi_2$}
                    \ENDIF
                    \IF{$\hd{\phi} = r$}
                        \STATE{$\Theta = \Theta \sqcup [(w, \tl{\phi})]$}
                    \ENDIF
                    \IF{($r \in S$ \textbf{and} $(w,h,r) \not\in E$)}
                        \IF{$\hd{\phi} = \overline{r}$}
                            \STATE{$\Theta = \Theta \sqcup [(w, \tl{\phi})]$}
                        \ENDIF
                    \ENDIF
                \ENDFOR
                \STATE{\COMMENT{consider edges directed towards $h$}}
                \FOR{each edge $(w,h,r) \in E$}
                    \IF{$\phi = \pi_1^* \comp \pi_2$}
                        \STATE{$\Theta = \Theta \sqcup [(h, \pi_2)]$}
                        \STATE{$\phi = \pi_1^+ \comp \pi_2$}
                    \ENDIF
                    \IF{$\hd{\phi} = \overline{r}$}
                        \STATE{$\Theta = \Theta \sqcup [(w, \tl{\phi})]$}
                    \ENDIF
                    \IF{($r \in S$ \textbf{and} $(h,w,r) \not\in E$)}
                        \IF{$\hd{\phi} = r$}
                            \STATE{$\Theta = \Theta \sqcup [(w, \tl{\phi})]$}
                        \ENDIF
                    \ENDIF
                \ENDFOR
                \STATE{\COMMENT{determine match or other nodes to visit}}
                \FOR{each $(n, \phi_s) \in \Theta$}
                    \IF{$(n, \phi_s) \not\in \seen$}
                        \IF{$\phi_s = \diamond$}
                            \IF{$n = v$}
       		                   \RETURN{\TRUE} \COMMENT{match}
                            \ENDIF
                        \ELSE
                            \STATE{\textbf{add} $(n, \phi_s)$ \textbf{to} $Q$}
                            \STATE{$\seen = \seen \cup \set{(n, \phi_s)}$}
                        \ENDIF
                    \ENDIF
                \ENDFOR
            \ENDWHILE
            \RETURN{\FALSE} \COMMENT{no match}
    \end{algorithmic}
\end{algorithm}

The algorithm uses a queue $Q$ to track nodes that we have to visit.
Unlike a conventional breadth-first search, we allow those nodes to be revisited because path conditions may be satisfied by a cycle in the system graph.
However, if we revisit a node then we require a different non-empty path condition on each visit.
In this way we avoid infinite loops whilst processing the path condition.

Previously visited nodes, and the path condition at the time of the visit, are tracked using the set $\seen$.
At each node $h$ we visit, we identify incident edges from our system graph $G = (V,E)$, it is these edges that we attempt to traverse by matching a label to the head of our path condition.
Matched edges result in path condition suffixes relevant at specific adjacent nodes; we hold these in a list as node-suffix pairs. We use the notation $list_1 \sqcup list_2$ to indicate the concatenation of two lists, where the entries from $list_2$ are appended, in order, to the end of $list_1$.

The algorithm performs its edge condition matching in lines 8 to 22 and 23 to 37 for outgoing and incoming edges to the current node $h$ respectively.
However, the implementation is complicated by the handling of path conditions of the form $\pi^+$, whose suffix includes $\pi^*$ (see Proposition~\ref{pro:path-condition-equiv-head-comp-suffix}), which may represent $0$ occurrences of $\pi$ or at least one occurrence of $\pi$.
When processing path conditions, therefore, we first determine if it has the structure $\pi_1^* \comp \pi_2$ (where $\pi_2$ may be $\diamond$) and, if so, we treat it as $\pi_1^+ \comp \pi_2$; in addition we add $\pi_2$ (corresponding to $0$ occurrences of $\pi_1$), along with the current node $h$, to our list of node-suffix pairs for consideration later (see Algorithm~\ref{alg:MatchPrincipal} lines 10 to 13 and 25 to 28).

After an edge condition is matched by the algorithm, each node-suffix pair is checked against $\seen$ and ignored if previously processed. The suffix $\phi_s$ of each unseen tuple is compared to $\diamond$, those matching indicate fully processed path conditions. If the node $n$ associated with such a tuple is equal to the target node $v$ then the path condition is considered to have been matched between $u$ and $v$ and the algorithm returns true. If the node isn't the target node, then the tuple is discarded as there is no remaining path condition to evaluate. Those unseen tuples, whose suffixes are not $\diamond$, are added to the queue of nodes to be visited (see Algorithm~\ref{alg:MatchPrincipal} lines 38 to 50).

Once all incident edges are considered for the current node, we move to the next node as indicated by the next entry in the queue $Q$. If the queue is empty and we have not already returned a value, then the path condition cannot be matched (because there are no further nodes to examine) and the algorithm returns false (see Algorithm~\ref{alg:MatchPrincipal} lines 5, 6 and 52).

Consider, for example, the system graph depicted in Figure~\ref{img:principal_matching_rule} and the path condition $r_1^+ \comp \overline{r_2} \comp r_3 \comp r_4$ with start node $s$ and end node $o$.
Then we are able to match edge condition $r_1$ and progress to node $v_1$ with path condition $r_1^* \comp \overline {r_2} \comp r_3 \comp r_4$.
We now attempt to match $r_1$ again, which fails.
In addition, we add $(v_1,\overline{r_2} \comp r_3 \comp r_4)$ to the list of node-path condition pairs to consider.
This will, eventually, lead to the node-suffix pair $(o,\diamond)$ being identified, at which point the algorithm will return a match (for path condition $r_1^+ \comp \overline{r_2} \comp r_3 \comp r_4$ with start and end nodes $s$ and $o$).

\subsection{Correctness and Complexity}\label{sec:Algorithm:Complexity}

We first introduce the concept of the length of a simple path condition.
Informally, it is equal to the number of edge conditions ($r$ or $\overline{r}$) which it contains.

\begin{definition}
    The \emph{length} $\ell(\pi)$ of simple path condition $\pi$ is defined as follows:
    \begin{itemize}
        \item  $\ell(r) = \ell(\overline{r}) = 1$;
        \item  $\ell(\pi \comp \pi') = \ell(\pi) + \ell(\pi')$;
        \item  $\ell(\pi^+) = \ell(\pi)$.
    \end{itemize}
    The \emph{length} $\ell(\rho)$ of a principal-matching policy $\rho$ is equal to the length of the longest path condition of the principal-matching rules within $\rho$,
    $\ell(\rho) = \max\limits_{\pi \in \rho} \left(\ell(\pi)\right)$.
\end{definition}

Our algorithm terminates because, with one exception discussed below, $\ell(\tl{\pi}) = \ell(\pi) - 1$ because an edge is consumed in matching the head of the path condition.
Thus, any node-suffix pair that is enqueued contains a shorter path condition.
Eventually, the path condition will be reduced to $\diamond$ and we test whether the adjacent node is the target node.
The exception arises when we consider a path condition of the form $\pi^*\comp \pi'$.
In this case, we enqueue a path condition of the form $\pi'$ and also evaluate the path condition $\pi^+ \comp \pi'$.
Thus, we could, in the worst case visit every node in the system graph and evaluate the path condition $\pi^+\comp\pi'$.
However, we do not enqueue a node-suffix pair if we have previously evaluated it (in the same way that a normal breadth-first search keeps track of visited nodes).
Thus, for this exceptional case, we will eventually process the path condition $\pi'$ (since $\pi^+\comp\pi'$ will either be discarded or fail to find a matching edge).

We can summarise our path-matching algorithm's processing as a breadth-first search through the graph, attempting to match edges to the remaining path condition.
At each node the number of possible comparisons depends on the degree of that node.
The path condition under consideration is re-written as each edge comparison is performed, with the head of the path condition removed if the edge satisfied the next element in the path condition.
The path condition under consideration at adjacent nodes is, therefore, one element shorter than at the current one.

The time complexity of a standard breadth-first search is determined by the number of nodes and edges, since, in the worst case, each node and edge will be explored.
For the \textsf{PrincipalMatch} algorithm, the number of ``nodes'' is determined by the number of nodes in the system graph and the length of the path condition.
Specifically, the size of the queue is bounded by $|V| \cdot \ell(\phi)$.
The number of edges in the system graph is $O(|V|^2 \cdot |R|)$.
Thus, the total complexity of the algorithm is $O(|V| \cdot \ell(\phi) + |V|^2 \cdot |R|)$.
%

The \textsf{MatchPrincipal} algorithm determines whether a single path condition matches.
In order to compute the list of matching principals in the worst case, every rule in the principal-matching policy $\rho$ may need to be evaluated.
The worst-case time complexity of principal matching is, therefore, determined by the complexity of matching one rule, the number of rules in the policy and $\ell(\rho)$.

\subsection{Implementation}
We have created a Python implementation of the \textsf{MatchPrincipal} algorithm which roughly follows the structure shown in Algorithm~\ref{alg:MatchPrincipal}.
We represent a path condition as a tree of nodes, where each node is a data structure containing
\begin{inparaenum}[(i)]
  \item pointers to a left and a right node
  \item a relationship label if it is a leaf node
  \item a node type if the node is a non-leaf node (indicating the operation used to construct the path condition).
\end{inparaenum}
Our implementation modifies the pseudo-code listed in Algorithm~\ref{alg:MatchPrincipal} in order to improve the processing of path conditions containing $\pi^*$.
In particular, we process both possibilities for $\pi^*$ when we meet it, rather than simply putting one aside for consideration later (as we do in lines 11 and 26 of Algorithm~\ref{alg:MatchPrincipal}).

%

Using this implementation we evaluated the requests in our appendix example.
The results are summarized in Table~\ref{tbl:implementation_metrics}, which shows the number of nodes visited ($n$) and edges considered ($e$) during the evaluation of one specific principal-matching rule for each of these requests.

\begin{table*}[!ht]\centering
  {\renewcommand{\arraystretch}{1.25}
  \begin{tabular}{|r|r|l|r|r|r|}
    \hline
        \bf Path condition $\pi$ & $\ell(\pi)$ & \bf Request & $n$ & $e$ & \bf Found\\
    \hline
    \hline
        $\textsf{P} \comp \overline{\textsf{R}} \comp \overline{\textsf{M}}^+$ & 3 & $(\textsf{Sales.\#2}, \textsf{Func.Spec.\#1}, \textsf{write})$ & 5 & 19 & Yes\\
        $\textsf{P} \comp \overline{\textsf{R}} \comp \overline{\textsf{M}}^+$ & 3 & $(\textsf{Tech.\#2}, \textsf{Test.Spec.\#1}, \textsf{read})$ & 7 & 24 & Yes\\
        $\textsf{S} \comp \overline{\textsf{R}} \comp \overline{\textsf{M}}^+$  & 3 & $(\textsf{Tech.\#2}, \textsf{Func.Spec.\#1}, \textsf{write})$ & 4 & 15 & Yes\\
        $\textsf{S}^+ \comp \overline{\textsf{M}} \comp \textsf{S} \comp \overline{\textsf{D}} \comp \overline{\textsf{M}}^+$ & 5 & $(\textsf{CTO}, \textsf{Proj.\#1 Report\#1}, \textsf{read})$ & 17 & 58 & Yes\\
        $\textsf{S}^+ \comp \overline{\textsf{M}} \comp \textsf{S} \comp \overline{\textsf{D}} \comp \overline{\textsf{M}}^+$ & 5 & $(\textsf{CEO}, \textsf{Proj.\#1 Report\#1}, \textsf{read})$ & 7 & 24 & No\\
    \hline
  \end{tabular}}
\caption{Running our implementation of \textsf{MatchPrincipal} using path conditions and requests from Tables~\ref{tbl:example_company_principal_matching_policy} and~\ref{tbl:example_company_requests}}\label{tbl:implementation_metrics}
\end{table*}

Notice that the algorithm may visit many more nodes than exist on the shortest path between the subject and object of the request.
This is because we are using a breadth-first search.
Notice also that two different subject nodes may be the same distance from the object node (as is the case for the subjects in the first and second rows) and yet one request is resolved with less computational effort.
It would be interesting to see whether there is any advantage to be gained in using a depth-first search.
This is certainly something we hope to investigate in future work.

\section{Related Work}\label{sec:RelatedWork}

We have already noted those aspects of the Unix access control model and role-based access control that have influenced the design of our model.
Our work also takes inspiration from the formal model developed for Unix by Crampton~\cite{Crampton_UnixAccessControl}, which suggested that the two-stage evaluation process used by Unix could provide inspiration for novel relationship-based access control models.
We now compare our model with related work in the literature.



The widespread use of social networks and restricting the access to resources within such networks has inspired the development of research into relationship-based access control.
The early work of Kruk {\em et al.} used {\sf friend} and {\sf friend-of-a-friend} relationships to determine access to resources~\cite{KrGrGzWoCh06}, while the work of Ali {\em et al.} was based on the trust relationships between users~\cite{AlViMa07}.
Carminati {\em et al.} synthesized these elements to create an access control model for social networks based on relationships~\cite{Carminati_Enforcing}.
They represent a social network as a graph in which the edges are labelled by relationships (such as {\sf friend}) and all nodes represent users.
Each edge is also labelled with a trust value, indicating the ``strength'' of the relationship.
An access condition has the form $(u,r,d,t)$, where $u$ is a user, $r$ is a relationship label, $d$ is the depth and $t$ is the trust threshold.
An access rule has the form $(o,C)$, where $o$ is an object and $C$ is a set of access conditions.
A user $v$ is authorized to access the resource $o$ if $v$ satisfies the access conditions specified in $C$.
More recent work has built on this model to provide additional features, such as joint management of access policies, but only in the context of social networks~\cite{HuAhJo13}.

If we ignore the trust threshold, access conditions are a special case of path conditions.
Specifically a relationship $r$ of depth $d$ can be represented by the path condition $r \comp \dots \comp r$ (repeated $d$ times).
Moreover, we can specify relationships of unbounded depth using the path condition $r^+$.
However, access conditions certainly cannot represent arbitrary path conditions.
In other words, our approach significantly extends the possibilities for policy specification.
(Trust thresholds may be useful in social networks, but we feel their use for our intended applications is inappropriate.
Of course, our framework may be easily adapted to accommodate trust thresholds by having a path condition built from pairs of the form $(r,t)$, where $r$ is a relationship label and $t$ is a threshold.)

Fong's recent work on relationship-based access control also concentrates on access control in social networks and models the social network as a graph in which the edges are labelled by relationships and all nodes are users~\cite{Fong_ReBAC,Fong_ESORICS09}.
Fong's work specifies a policy for each resource, where a policy is specified using a multi-modal logic.\footnote{The rationale for using a modal logic is that each relationship specifies an accessibility relation between users, which is used to provide semantics for policies.  We do this more directly by working with path conditions and specifying their semantics in terms of a graph.}
Thus Fong's work provides a richer policy language than that of Carminati {\em et al.}
The policy syntax is specified by the grammar \[  \phi, \psi ::= \top \mid \textsf{a} \mid \neg \phi \mid \phi \vee \psi \mid \langle i \rangle \phi \] where $i$ is a relationship identifier.
Informally, $\top$ serves the same purpose as our default path condition $\top$; $\sf a$ is analogous to $\diamond$; $\langle i \rangle$ is equivalent to our path condition $r$.
Fong's language can encode alternatives (using $\vee$); we would simply specify alternative principal-matching rules.
Fong's language does support negation, which we do not.
Conversely, our language does support unbounded path conditions, which are useful when traversing a sub-graph comprising similar types of elements that might have arbitrary diameter (as in a directory tree, for example).
A limitation of Fong's language is that a policy has to be specified for every resource and admits no relationships, other than ownership, between users and resources.
In our approach, we simply identify the principals that apply to a request, given the subject and object of the request, thereby leveraging some of the advantages of a role-based approach.
Moreover, we allow for arbitrary relationships between the nodes (subject to the constraints in the permissible relationship graph) in the system graph.

Cheng {\em et al.} also focused on the use of relationship-based access control within Online Social Networks ~\cite{ChengPS12passat,ChengPS12dbsec}.
Their work allows for the specification of user-to-resource relationships (other than ownership).
However, our model is more general still in its support for entities of any kind (including logical ones) and policies not focused on, but still applicable to, social networks.
Cheng {\em et al.} employ a path checking algorithm which is comparable to our concept of path matching.
However, their approach directly assigns permissions, whereas we introduce some of the benefits of RBAC and Unix access control by abstracting that assignment to matched principals.
Their path expressions are directly based on regular expressions, including wildcards, although they constrain rules containing wildcards so that such rules could, in fact, be enumerated as different alternatives.
(Thus paths of arbitrary length are not properly supported.)
In contrast, we only provide direct support for $\pi^+$ in path conditions, but do not limit the number of edges across which it can match, something that is crucial when dealing with variable depth data structures such as directory trees.
Moreover, as we have seen, we can encode alternation in a rule's path condition as two (or more rules): the rule $(\pi_1 \mid \pi_2,p)$ is simply defined as two principal-matching rules $(\pi_1,p)$ and $(\pi_2,p)$.
Similarly, we can handle $(\pi^*,p)$ by defining the rules $(\pi^+,p)$ and $(\diamond,p)$.

\section{Conclusion}\label{sec:Conclusion}

We have formally defined a new graph-based model for access control based on two concepts: path conditions and principal matching.
We believe that path conditions are a novel contribution to the literature on relationship-based access control and that these conditions allow us to specify a wide range of policies that are relevant to access control in a wide range of applications, not just in the usual context of social networks.
Principal matching enables us to leverage the advantages of both Unix and RBAC and extend the capabilities of both models.
We also believe our model provides significant advantages over existing models for relationship-based access control, both in terms of the expressive power of path conditions and the relatively straightforward request evaluation process.
Additionally, our model is generic, thus able to describe systems of various forms be they social networks, IT systems (singularly or as networks) or entire businesses.
We have illustrated how the model can be implemented by describing an algorithm to support principal matching and, thereby, enable request evaluation within our model.

There are many opportunities for further work.
In particular, we would like to investigate alternative path-matching algorithms and compare their efficiency with the one described in Section~\ref{sec:Algorithm}.
SPARQL is an RDF query language that may well be an suitable alternative.
We would also like to extend the policy language to include more expressive matching as a means of directly supporting access constraints such as separation of duty, binding of duty and Chinese Wall.
We believe that such constraints can be supported simply by introducing conjunction within, and negation of, path conditions.
Extending this further we plan to consider the matching of subgraphs, rather than paths, and to investigate the trade-offs in increased expressive power with the more expensive request evaluation algorithms that will be required.
We also intend to develop an administrative model to manage components such as the system graph.
In this way, we should be able to handle dynamic concepts, such as sessions in RBAC.
RT is a family of role-based trust management languages~\cite{LiMiWi02} that combine features of RBAC with distributed access control models.
Many of the rules of RT can, like the assignment relations in RBAC, be encoded as a single type of relationship within a system graph in our model.
However, the RT delegation rule $A.r \leftarrow A'.r'.r''$ cannot be directly encoded within our model.
We would like to be able to provide support for distributed access control, in which different parts of the subgraph form different administrative domains.
Then RT-like rules would specify the edges that link different subgraphs.
Finally, we would also like to enrich the model with stateful objects, such as workflow tasks, for which the set of authorized individuals may change over time.
We expect that this will result in the system graph being updated as the state of an object changes (for example to support task-based separation of duty).

\bibliography{../RelixAC}
\bibliographystyle{abbrv}

\clearpage

\appendix
\section{Corporate Example}\label{appendix:section:example}

The following example applies our model to the project environment within a fictional company. To support this specific system, we initially define the underlying system model $(T,R,S,G_{\textrm{PR}})$ where the set of types is
\[T = \{\textsf{File}, \textsf{Folder}, \textsf{Group}, \textsf{Printer}, \textsf{Project}, \textsf{User}\}\]
The set of relationship labels used in the system model is
\[
    \begin{array}{l}
        R = \{\textsf{\textbf{C}lient-of}, \textsf{\textbf{D}eliverable-for}, \textsf{\textbf{M}ember-of}, \\
        \textsf{\textbf{P}articipant-of}, \textsf{\textbf{R}esource-for}, \textsf{\textbf{S}upervises}\}
    \end{array}
\]
There are no symmetric relationship labels.
Finally, the permissible relationship graph $G_{\textrm{PR}}$, defined using $T$ and $R$, is shown in Figure~\ref{img:example_company_edge_constraints}.

\begin{figure}[!ht]
  \centering
  \includegraphics[scale=.55]{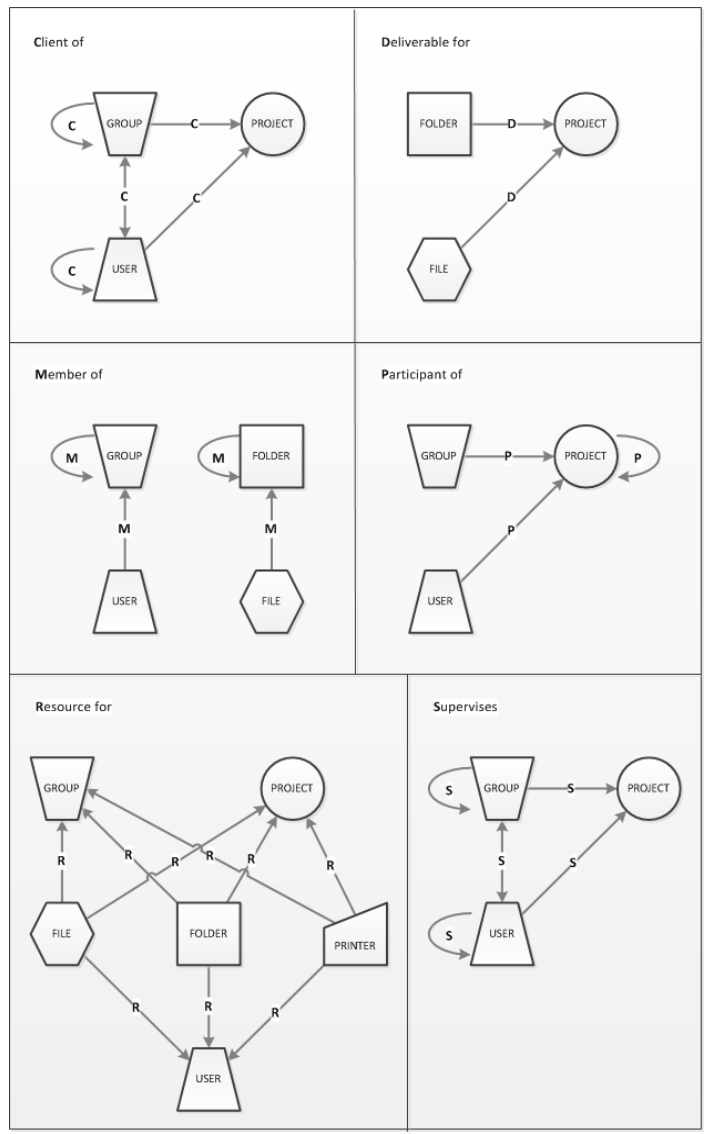}
  \caption{Permissible relationship graph}\label{img:example_company_edge_constraints}
\end{figure}

Using this system model we then describe the project environment using the system graph shown in Figure~\ref{img:example_company_system_graph}.

\begin{figure*}[!ht]
  \centering
  \includegraphics[scale=.55]{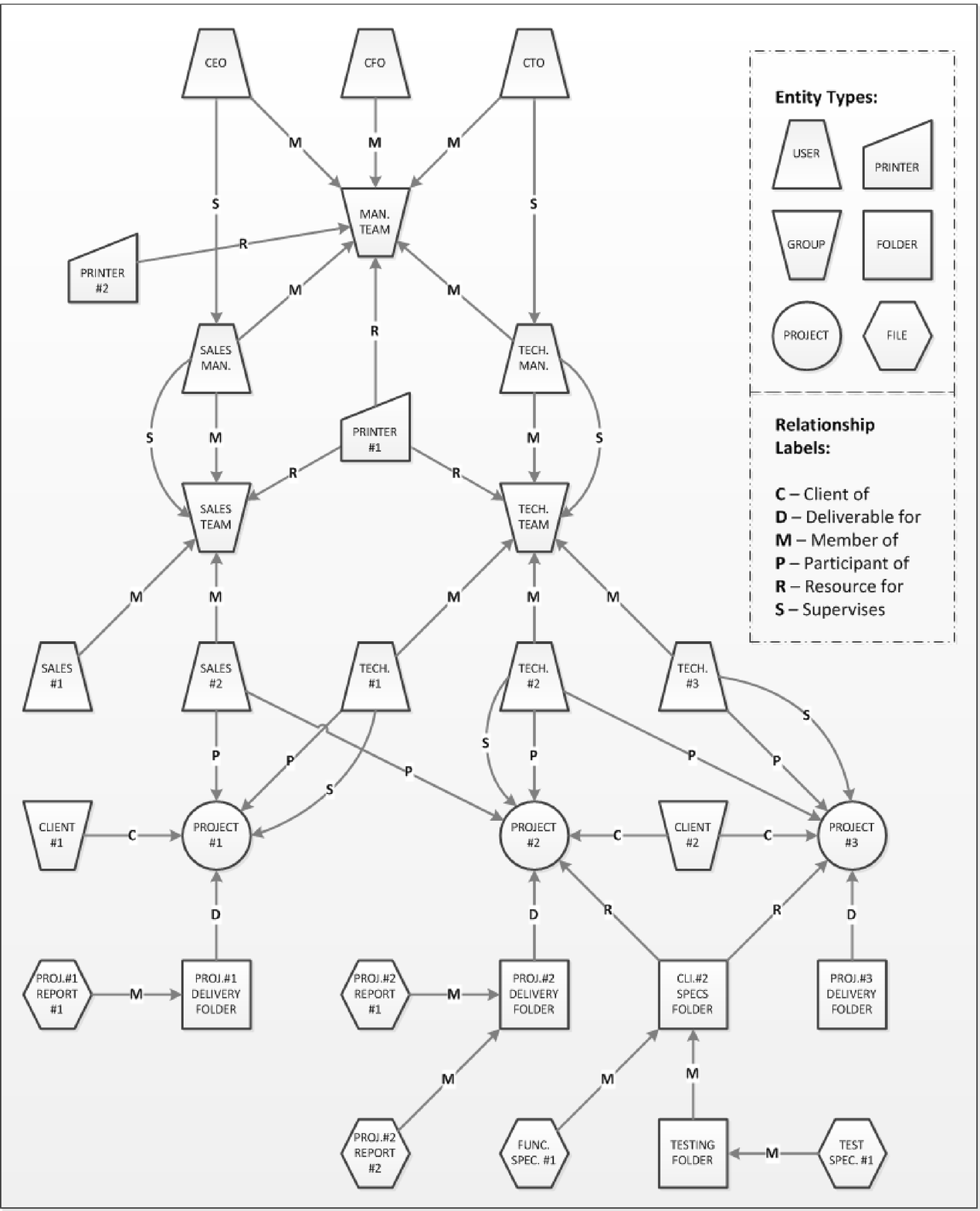}
  \caption{System graph}\label{img:example_company_system_graph}
\end{figure*}

Within our authorization system we define the principal-matching rules shown in Table~\ref{tbl:example_company_principal_matching_policy} and make use of the \textsf{AllMatch} PMS.

\begin{table}[!ht]\centering
  {\renewcommand{\arraystretch}{1.25}
  \begin{tabular}{|r|l|}
    \hline
        \bf \# & \bf Principal-Matching Rule \\
    \hline
    \hline
        1 & $(\textsf{C} \comp \overline{\textsf{D}} \comp \overline{\textsf{M}}^+, \textsf{Deliverable Client})$ \\
        2 & $(\textsf{S}^+ \comp \overline{\textsf{M}} \comp \textsf{S} \comp \overline{\textsf{D}}, \textsf{Deliverable Reviewer})$ \\
        3 & $(\textsf{S}^+ \comp \overline{\textsf{M}} \comp \textsf{S} \comp \overline{\textsf{D}} \comp \overline{\textsf{M}}^+, \textsf{Deliverable Reviewer})$ \\
        4 & $(\textsf{S} \comp \overline{\textsf{D}}, \textsf{Deliverable Supervisor})$ \\
        5 & $(\textsf{S} \comp \overline{\textsf{D}} \comp \overline{\textsf{M}}^+, \textsf{Deliverable Supervisor})$ \\
        6 & $(\textsf{P} \comp \overline{\textsf{D}}, \textsf{Deliverable User})$ \\
        7 & $(\textsf{P} \comp \overline{\textsf{D}} \comp \overline{\textsf{M}}^+, \textsf{Deliverable User})$ \\
        8 & $(\textsf{S} \comp \overline{\textsf{R}}, \textsf{Project Resource Supervisor})$ \\
        9 & $(\textsf{S} \comp \overline{\textsf{R}} \comp \overline{\textsf{M}}^+, \textsf{Project Resource Supervisor})$ \\
        10 & $(\textsf{P} \comp \overline{\textsf{R}}, \textsf{Project Resource User})$ \\
        11 & $(\textsf{P} \comp \overline{\textsf{R}} \comp \overline{\textsf{M}}^+, \textsf{Project Resource User})$ \\
        12 & $(\textsf{M} \comp \overline{\textsf{R}}, \textsf{Team Resource User})$ \\
    \hline
  \end{tabular}}
\caption{Principal-matching policy}\label{tbl:example_company_principal_matching_policy}
\end{table}

Additionally, we define the authorization policy shown in Table~\ref{tbl:example_company_authorization_policy} and whilst we define no per-subject or per-object defaults, we define the system-wide default as \textsf{deny-by-default}. We employ the \textsf{FirstMatch} conflict resolution strategy.

\begin{table}[!ht]\centering
  {\renewcommand{\arraystretch}{1.25}
  \begin{tabular}{|r|l|}
    \hline
        \bf \# & \bf Authorization Rule \\
    \hline
    \hline
        1 & $(\textsf{Deliverable Client}, \star, \textsf{read}, 1)$ \\
        2 & $(\textsf{Deliverable Reviewer}, \star, \textsf{read}, 1)$ \\
        3 & $(\textsf{Deliverable Supervisor}, \star, \textsf{read}, 1)$ \\
        4 & $(\textsf{Deliverable Supervisor}, \star, \textsf{write}, 1)$ \\
        5 & $(\textsf{Deliverable User}, \star, \textsf{read}, 1)$ \\
        6 & $(\textsf{Project Resource Supervisor}, \star, \textsf{read}, 1)$ \\
        7 & $(\textsf{Project Resource Supervisor}, \star, \textsf{write}, 1)$ \\
        8 & $(\textsf{Project Resource User}, \star, \textsf{read}, 1)$ \\
        9 & $(\textsf{Project Resource User}, \textsf{Func.Spec.\#1}, \textsf{write}, 0)$ \\
        10 & $(\textsf{Project Resource User}, \star, \textsf{write}, 1)$ \\
        11 & $(\textsf{Team Resource User}, \star, \textsf{write}, 1)$ \\
    \hline
  \end{tabular}}
\caption{Authorization policy}\label{tbl:example_company_authorization_policy}
\end{table}

Table~\ref{tbl:example_company_requests} lists some illustrative requests, together with the result of their evaluation.
Requests 1 and 2 would result in the list of matched principals $[\textsf{Project Resource Supervisor}, \textsf{Project Resource User}]$.
Requests 3 and 4 would result in the matched principal lists $[\textsf{Project Resource User}]$ and $[\textsf{Deliverable reviewer}]$, respectively, while the final request would match no principals.

\begin{table*}[!ht]\centering
  {\renewcommand{\arraystretch}{1.25}
  \begin{tabular}{|r|l|l|l|l|}
    \hline
        \bf \# & \bf Request & \bf Decision set & \bf Outcome & \bf Comment \\
    \hline
    \hline
        1 & $(\textsf{Tech.\#2}, \textsf{Test.Spec.\#1}, \textsf{read})$ & $\set{1}$ & Allow & \\
        2 & $(\textsf{Tech.\#2}, \textsf{Func.Spec.\#1}, \textsf{write})$ & $\set{1,0}$ & Allow & First match \\
        3 & $(\textsf{Sales.\#2}, \textsf{Func.Spec.\#1}, \textsf{write})$ & $\set{0}$ & Deny & \\
        4 & $(\textsf{CTO}, \textsf{Proj.\#1 Report\#1}, \textsf{read})$ & $\set{1}$ & Allow & \\
        5 & $(\textsf{CEO}, \textsf{Proj.\#1 Report\#1}, \textsf{read})$ & $\set{}$ & Deny & System default \\
    \hline
  \end{tabular}}
\caption{Sample requests}\label{tbl:example_company_requests}
\end{table*}

\end{document}